\pdfoutput=1 
\documentclass{article}

\usepackage[margin=1in]{geometry}
\usepackage{amsfonts,amsmath,amssymb,amsthm,mathtools,stmaryrd}

\usepackage{thm-restate}

\usepackage{microtype,xspace,bm,doi}
\usepackage{hyperref}
\usepackage[dvipsnames]{xcolor}
\hypersetup{
	colorlinks=true,
	urlcolor=MidnightBlue,
	linkcolor=MidnightBlue,
	citecolor=MidnightBlue,
	unicode
}
\usepackage[capitalise,nameinlink,noabbrev]{cleveref}
\crefname{equation}{}{}
\usepackage{crossreftools}
\usepackage{tikz}
\usetikzlibrary{shapes.arrows, fadings}

\let\ORGhypersetup\hypersetup
\protected\def\hypersetup{\ORGhypersetup}
\theoremstyle{definition}

\theoremstyle{plain}
\newtheorem{theorem}{Theorem}
\newtheorem{lemma}{Lemma}

\newtheorem{claim}{Claim}

\theoremstyle{remark}
\newtheorem{remark}{Remark}

\usepackage{enumitem}
\setlist{listparindent=\parindent,parsep=0pt,itemsep=1em}
\setlist[itemize]{label=$-$,noitemsep}
\setlist[enumerate]{itemsep=1mm}
\setlist[description]{leftmargin=\parindent}

\newcommand{\Xor}{\text{\upshape\textsc{Xor}}}
\newcommand{\Or}{\text{\upshape\textsc{Or}}}

\newcommand{\x}{{\bm{x}}}
\newcommand{\y}{{\bm{y}}}
\newcommand{\ib}{{\bm{i}}}
\newcommand{\Ub}{{\bm{U}}}
\newcommand{\Ab}{{\bm{A}}}
\newcommand{\R}{{\bm{R}}}
\newcommand{\T}{{\bm{T}}}
\newcommand{\X}{{\bm{X}}}
\newcommand{\Y}{{\bm{Y}}}

\newcommand{\As}{\mathcal{A}}

\renewcommand{\epsilon}{\varepsilon}
\renewcommand{\setminus}{\smallsetminus}

\newcommand{\bool}{\{0, 1\}}

\DeclareMathOperator*{\dist}{dist}
\DeclareMathOperator*{\supp}{supp}
\DeclareMathOperator*{\Exp}{\mathbb{E}}
\let\Pr\undefined
\DeclareMathOperator*{\Pr}{Pr}
\newcommand{\entropy}{\operatorname{H}}

\let\OLDthebibliography\thebibliography
\renewcommand\thebibliography[1]{
  \OLDthebibliography{#1}
  \setlength{\parskip}{2pt}
}

\begin{document}

\mbox{}\vspace{8mm}

\begin{center}
{\huge Top-Down Lower Bounds for Depth-Four Circuits}
\\[1.3cm] \large
	
\setlength\tabcolsep{1.2em}
\begin{tabular}{cccc}
Mika G\"o\"os&
Artur Riazanov&
Anastasia Sofronova&
Dmitry Sokolov\\[-1mm]
\small\slshape EPFL &
\small\slshape EPFL &
\small\slshape EPFL &
\small\slshape EPFL
\end{tabular}

\vspace{6mm}
	
\normalsize

{\today}
	
\vspace{4mm}
\end{center}

\begin{quote}
\noindent\small
{\bf Abstract.}~
We present a top-down lower-bound method for depth-$4$  boolean circuits. In particular, we give a new proof of the well-known result that the parity function requires {depth-$4$} circuits of size exponential in $n^{1/3}$. Our proof is an application of robust sunflowers and block unpredictability.
\end{quote}

\section{Introduction}

The working complexity theorist has three main weapons in their arsenal when proving lower bounds against small-depth boolean circuits (consisting of $\land$, $\lor$, $\neg$ gates of unbounded fanin). The most wildly successful ones are the \emph{random restriction} method~\cite{Furst1984,Ajtai1983} and the \emph{polynomial approximation} method~\cite{Razborov1987,Smolensky1987}. The random restriction method, in particular, is applied {\bf\itshape bottom-up}: it starts by analysing the bottom-most layer of gates next to input variables and finds a way to simplify the circuit so as to reduce its depth by one. The third main weapon, which is the subject of this paper, is the {\bf\itshape top-down} method: starting at the top (output) gate we walk down the circuit in search of a mistake in the computation.
Top-down methods (often phrased in the language of communication complexity~\cite{Karchmer1990}) are routinely used to prove:
\begin{enumerate}
\item Polynomial lower bounds for formulas~\cite{Karchmer1995,Edmonds2001,Gavinsky2017,Dinur2017,Rezende2020}.
\item Exponential lower bounds for \emph{monotone} circuits and formulas~\cite{Karchmer1990,Raz1992,Nisan1993,Grigni1995,Raz1999,Goos2018,Rezende2016,Pitassi2017,Garg2020,Rezende2022}.
\item \label{item}
Exponential lower bounds for depth-3 circuits~\cite{Baker1979,Santha1989,Ko1990,Hastad1995,Russell1998,Paturi1999,Paturi2000,Impagliazzo2001,Paturi2005,Wolfovitz2006,Bohler2006,Meir2019,Golovnev2021,Frankl2022,Goos2023}.
\end{enumerate}

In particular, it has been an open problem (posed in~\cite{Hastad1995,Meir2019}) to prove exponential lower bounds for depth-4 circuits by a top-down argument. We develop such a lower-bound method in this paper and use it to prove a lower bound for the parity function. It has been long known using bottom-up methods that the depth-4 complexity of $n$-bit parity is $\smash{2^{\Theta(n^{1/3})}}$~\cite{Yao1985,Hastad1987}. We recover a slightly weaker bound.
\begin{theorem} \label{thm:main}
Every depth-$4$ circuit computing the $n$-bit parity requires $2^{n^{1/3-o(1)}}$ gates.
\end{theorem}

Our top-down proof of this theorem is a relatively simple application of two known techniques: robust sunflowers~\cite{Rossman2014,Alweiss2021,Rao2020} and unpredictability from partial information~\cite{Meir2019,Smal2018,Viola2021}, which we generalise to blocks of coordinates (obtaining essentially best possible parameters).

A major motivation for the further development of top-down methods is that the method is, in a precise sense, \emph{complete} for constant-depth circuits, in that it can be used to prove tight lower bounds (up to polynomial factors) for \emph{any} boolean function; see \cref{rem:complete}. The same is not known to hold for the aforementioned bottom-up techniques. For example, there is currently no known bottom-up proof for the depth-3 circuit lower bound that underlies the oracle separation $\textsf{AM}\not\subseteq \Sigma_2\textsf{P}$~\cite{Santha1989,Ko1990,Bohler2006}. We suspect more generally that top-down methods could prove useful in settings where the bottom-up methods have failed so far, such as proving lower bounds against $\mathsf{AC}^0\circ\oplus$ circuits computing inner-product~\cite{Cheraghchi2018,Ezra2022,Huang2022,Servedio2012} or against the polynomial hierarchy in communication complexity~\cite{Babai1986}.

\section{Proof overview}

Before we explain our new top-down lower bound method for depth-4 circuits, let us review one particular top-down technique for depth-3 circuits that we build on. Namely, a lower bound using the information-theoretic \emph{unpredictability lemma} of Meir and Wigderson~\cite{Meir2019,Smal2018,Viola2021}.

\subsection{Depth-3 via bit unpredictability}
Suppose $X\subseteq \bool^n$ is a set of $n$-bit strings and $i\in[n]$ a coordinate. A \emph{certificate} for $i$ with respect to~$X$ is a pair $(Q,a)$ such that $Q \subseteq [n] \setminus \{i\}$, $a \in \{0,1\}^Q$, and there exists a bit $b \in \{0,1\}$ such that every $x \in X$ with $x_Q = a$ satisfies $x_i = b$. In words, the partial assignment defined by $(Q,a)$ will correctly predict that the $i$-th bit must be $b$.  The \emph{size} of the certificate is $|Q|$. We also say that $x$ \emph{contains} a size-$q$ certificate for $i$ (wrt~$X$) if there is some size-$q$ certificate of the form $(Q,x_Q)$ for $i$. Meir and Wigderson~\cite{Meir2019} proved that if $X$ is a large set, a uniform random string $\x\sim X$ will not contain a small certificate for a uniform random coordinate $\ib\sim[n]$ with high probability. There is also an alternative proof by Smal and Talebanfard~\cite{Smal2018} which tightens the parameters.
\begin{lemma}[Bit unpredictability~{\cite{Meir2019}}]
\label{lem:bit-unpred}
Let $X\subseteq\{0,1\}^n$ have density $|X|/2^n\geq 2^{-k}$. Then for any $q\geq 1$,
\[
\Pr_{(\x,\ib)\sim X\times  [n]}[\, \text{$\x$ contains a size-$q$ certificate for $\ib$ wrt $X$}\,] 
~\leq~
O(kq/n).
\]
\end{lemma}

Let us apply this lemma to show a $2^{\Omega(\sqrt{n})}$ depth-3 lower bound for the $n$-bit parity function $\Xor$. Suppose for the sake of contradiction that $\Sigma$ is a depth-3 circuit of size $|\Sigma|=2^{o(\sqrt{n})}$ for $\Xor$. We may assume that $\Sigma$ is of type $\lor\circ\land\circ\lor$, that is, the circuit has an $\lor$-gate at the top, followed by a layer of $\land$-gates, then a layer of $\lor$-gates, and finally we have the input literals at the bottom. Assume for simplicity of exposition that the bottom fanin of $\Sigma$ is at most $\sqrt{n}$. (If we are allowed to invoke a bottom-up trick, then this bottom-fanin assumption can be easily ensured by restricting a small fraction of input variables~\cite[Lemma~3.2]{Hastad1995}.) That is, $\Sigma$ computes an $\Or$ of $\sqrt{n}$-CNFs. Because the top gate is $\lor$, we have $\Xor^{-1}(1)=\bigcup_{j\in [s]} \Pi_i^{-1}(1)$ where $s\leq|\Sigma|$ is the top fanin and $\Pi_j$ is the $j$-th CNF feeding into the top gate. We now take a naive greedy step down the circuit: we choose any $j$ that maximises $|\Pi_j^{-1}(1)|$ and set $\Pi\coloneqq\Pi_j$.
In summary,
\begin{itemize}
    \item $\Pi$ accepts the set~$X\coloneqq\Pi^{-1}(1)$ of density $|X|/2^n\geq |\Xor^{-1}(1)|/(2^n|\Sigma|)\geq 1/(2|\Sigma|) \geq 2^{-o(\sqrt{n})}$.
    \item $\Pi$ rejects the set $\Xor^{-1}(0)$.
\end{itemize}

We can now apply \cref{lem:bit-unpred} to the set $X$ with parameters $k\coloneqq o(\sqrt{n})$ and $q\coloneqq \sqrt{n}$. As a result, there exist some string $x^*\in X$ and a coordinate $i\in[n]$ such that $x^*$ does not contain any size-$q$ certificates for $i$ wrt~$X$. Consider the string $y\in\Xor^{-1}(0)$ that is obtained from $x^*$ by flipping the $i$-th bit. We claim that~$y$ is locally indistinguishable from strings in $X$---often $y$ is called a \emph{local limit} of $X$---in the following sense.%
\begin{claim}[Local limit] \label{clm:local}
For every $Q\subseteq[n]$, $|Q|\leq\sqrt{n}$, there exist an $x \in X$ such that $x_Q=y_Q$.
\end{claim}
\begin{proof}
Let $Q'\coloneqq Q\setminus\{i\}$. Since $(Q',x^*_{Q'})$ is not a certificate for $i$ wrt $X$, we get that, for the bit $b\coloneqq 1-x^*_i$, there exist some $x\in X$ such that $x_{Q'}=x^*_{Q'}$ and $x_i=b$. But this means $x_Q=y_Q$ (see~\cref{fig:enter-label}).
\end{proof}

We finally claim that $\Pi$ accepts $y\in\Xor^{-1}(0)$ which is the desired contradiction. To show this, we need to show that every clause of $\Pi$ accepts $y$. Consider any clause $\Lambda$ of $\Pi$ and let $Q\subseteq[n]$, $|Q|\leq\sqrt{n}$, be the set of variables mentioned in $\Lambda$. By \cref{clm:local} there is some $x\in X$ such that $\smash{x_Q=y_Q}$. But since $\Lambda$ accepts $x$, it must also accept $y$, as desired. This concludes the proof that $\Xor$ requires depth-3 circuits of size $2^{\Omega(\sqrt{n})}$.

\begin{figure}
    \centering
    \begin{tikzpicture}
    \def\elements{{
            {1, 1, 0, 1, 1},
            {0, 1, 0, 1, 0},
            {1, 1, 1, 1, 1},
            {0, 1, 1, 1, 0},
            {1, 1, 1, 0, 1},
            {0, 1, 0, 1, 1},
            {1, 1, 1, 0, 1},
            {0, 0, 1, 0, 0},
            {1, 1, 1, 1, 1}
        }}
    \def\result{{1, 0, 1, 1, 1, 1, 1, 0, 1}} 
    \def\n{8}
    \def\m{4}
    \def\step{0.65}

    \draw[step = \step, thick] (0, 0) grid ({\step * (\n + 1)}, {-\step * (\m + 1)});
    \draw[thick, blue!30!black, rounded corners = 2pt, fill = blue, fill opacity = 0.5]
        (3 * \step, {-\step * (\m + 1)}) rectangle (4 * \step, {-\step * \m});
    \draw[thick, green!30!black, rounded corners = 2pt, fill = green, fill opacity = 0.2]
        (-0.02, {-\step * (\m + 1) - 0.02}) rectangle ({\step * (\n + 1) + 0.02}, {-\step * \m + 0.02});

    \foreach \i in {0, 1, ..., \n}{
        \foreach \j in {0, 1, ..., \m}{
            \node at ({\step / 2 + \i * \step}, {-\step / 2 - \j * \step})
                {\pgfmathparse{\elements[\i][\j]}\pgfmathresult};
        }
    }

    \begin{scope}[shift = {(0, {-\step * (\m + 1) - 1})}]
        \draw[step = \step, thick] (0, 0) grid ({\step * (\n + 1)}, -\step);
        \draw[thick, blue!30!black, rounded corners = 2pt, fill = blue, fill opacity = 0.5]
            (3 * \step, 0) rectangle (4 * \step, -\step);
            
        \foreach \i in {0, 1, ..., \n}{
            \node at ({\step / 2 + \i * \step}, {-\step / 2})
                {\pgfmathparse{\result[\i]}\pgfmathresult};
        }
    \end{scope}

    \node at (-1.5, {-\step * (\m + 1) / 2}) {\Large $X \coloneqq$};
    \node at (-1.5, {-\step * (\m + 1.5) - 1}) {Local limit:};

    \node at ({\step * (\n + 1) + 2.8}, {-\step * (\m + 1) / 2}) {$\Pi$ accepts all $X$};
    \node at ({\step * (\n + 1) + 2.8}, {-\step * (\m + 1.5) - 1}) {$\Pi$ also accepts local limit};
\end{tikzpicture}
    \caption{Flipping one bit yields a local limit}
    \label{fig:enter-label}
\end{figure}

\subsection{Depth-4 via block unpredictability}

Our depth-4 lower bound will encounter a number of new challenges compared to the depth-3 proof above.

\paragraph{Challenge 1: Which subcircuit to pick?}
The first challenge is that we have to be more careful how to take steps down the circuit.
Suppose for contradiction $\Pi$ is a small depth-4 circuit of type $\land\circ\lor\circ\land\circ\lor$
computing $\Xor$. Our first step will still be naively greedy: we choose the depth-3 subcircuit
$\Sigma$ of $\Pi$ that rejects as large a set $Y\subseteq \Xor^{-1}(0)$ as possible. The second
step is trickier. Intuitively, we should choose a depth-2 subcircuit (namely, a CNF) of $\Sigma$
that accepts not just many 1-inputs but in particular such 1-inputs that are hard to distinguish
from $Y$. To this end, we first construct a set $M\subseteq\Xor^{-1}(1)$ that~\textbf{{\itshape mirrors}}
the set $Y\subseteq\Xor^{-1}(0)$ in the sense that $M$ is hard to locally distinguish from $Y$. Concretely,
let us define (as a first attempt; to be adjusted later) the mirror set $M$ as the set of $1$-inputs
$x$ such that the $n^{1/3}$-radius Hamming ball around $x$ contains many points from $Y$. Our second
step down the circuit will then be to choose any CNF $\Gamma$ that accepts a large fraction of the
mirror set $M$. We illustrate the process in~\cref{fig:steps}.

\begin{figure}
    \centering
    \tikzset{
    vert/.style = {
        circle,
        thick,
        minimum size = 0.5cm
    },
    fancy-arrow/.style = {
        draw = black,
        thick,
        single arrow,
        right color = green!40!black,
        left color = red!40!black,
        fill opacity = 0.7,
        single arrow head extend = 0.3cm,
        single arrow tip angle = 70,
        single arrow head indent = 0.15cm,
        minimum height = 1.4cm,
        minimum width = 1cm,
        shading angle = -45,
        rotate = -90
    },
    common-part/.pic = {
        \draw[thick] (0, 0) rectangle (\a, \a);
        \draw[thick] (0, 0) -- (\a, \a);

        \node at (0.5, \a - 0.5) {\Large $0$};
        \node at (\a - 0.5, 0.5) {\Large $1$};

        \foreach \i in {1, 2, ..., 5}{
            \draw[thick, ->] (a) -- (b\i);
        }
    },
}
    
\begin{tikzpicture}[>=latex]
    \def\a{4}

    \fill[red!70!black, opacity = 0.3] (0, 0) -- (\a, \a) -- (\a, 0) -- cycle;
    \fill[green!50!black, opacity = 0.3] (0, 0) -- (\a, \a) -- (0, \a) -- cycle;

    \begin{scope}[shift = {(9, 0)}]
        \node[draw, vert, color = green!40!black, fill = green!20] (a) at (0, \a - 0.5)
            {\textcolor{black}{$\land$}};

        \node[draw, vert] (b1) at (-3, \a - 2.5) {$\lor$};
        \node[draw, vert] (b2) at (-1.5, \a - 2.5) {$\lor$};
        \node[vert] (b3) at (0, \a - 2.5) {};
        \node[vert] (b4) at (1.5, \a - 2.5) {};
        \node[draw, vert] (b5) at (3, \a - 2.5) {$\lor$};

        \foreach \i in {1, 2, ..., 5}{
            \draw[thick, ->] (a) -- (b\i);
        }

        \draw[<->, thick, blue] (0, \a - 0.5) ++(-165:0.9) arc (-160:-15:0.9) coordinate (c);
        \draw (c) ++(0.5, 0) node {\textcolor{blue}{$2^{n^{\varepsilon}}$}};
    \end{scope}

    \pic {common-part};
    \node[fancy-arrow] at (6.25, -1) {};

    \begin{scope}[shift = {(0, -\a - 2.1)}]
        \fill[red!70!black, opacity = 0.3] (0, 0) -- (\a, \a) -- (\a, 0) -- cycle;
        \draw[green!50!black, thick, fill = green!50!black, fill opacity = 0.3]
            (1.1, 1.2) -- (2.2, 2.3) to[out = 45, in = 40] (1.3, 2.3) to[out = 220, in = 225] cycle;
        \node[green!50!black] at (1.3, 3) {\small $\mathrm{H}(Y) = n - n^{\varepsilon}$};
        \node[green!50!black] at (1.5, 2) {\small $Y$};
        \draw[->, green!50!black, thick] (0.5, 2.8) to[out = -90, in = 140] (1, 2.1);

        \begin{scope}[shift = {(9, 0)}]
            \node[draw, vert] (a) at (0, \a - 0.5) {$\land$};

            \node[draw, vert] (b1) at (-3, \a - 2.5) {$\lor$};
            \node[draw, vert, color = green!40!black, fill = green!20] (b2) at (-1.5, \a - 2.5)
                {\textcolor{black}{$\lor$}};
            \node[vert] (b3) at (0, \a - 2.5) {};
            \node[vert] (b4) at (1.5, \a - 2.5) {};
            \node[draw, vert] (b5) at (3, \a - 2.5) {$\lor$};

            \draw[<->, thick, blue] (0, \a - 0.5) ++(-165:0.9) arc (-160:-15:0.9) coordinate (c);
            \draw (c) ++(0.5, 0) node {\textcolor{blue}{$2^{n^{\varepsilon}}$}};
        \end{scope}

        \pic {common-part};
        \node[fancy-arrow] at (6.25, -1) {};
    \end{scope}

    \begin{scope}[shift = {(0, -2 * \a - 4.2)}]
        \draw[green!50!black, thick, fill = green!50!black, fill opacity = 0.3]
            (1.1, 1.2) -- (2.2, 2.3) to[out = 45, in = 40] (1.3, 2.3) to[out = 220, in = 225] cycle;
        \node[green!50!black] at (1.5, 2) {\small $Y$};

        \draw[red!70!black, thick, fill = red!70!black, fill opacity = 0.3]
            (1.1, 1) -- (2.2, 2.1) to[out = 45, in = 40] (2.3, 0.9) to[out = 220, in = 225] cycle;
        \node[red!70!black] at (1.6, 1) {\small $M$};

        \begin{scope}[shift = {(9, 0)}]
            \node[draw, vert] (a) at (0, \a - 0.5) {$\land$};

            \node[draw, vert] (b1) at (-3, \a - 2.5) {$\lor$};
            \node[draw, vert, color = green!40!black, fill = green!20] (b2) at (-1.5, \a - 2.5)
                 {\textcolor{black}{$\lor$}};
            \node[vert] (b3) at (0, \a - 2.5) {};
            \node[vert] (b4) at (1.5, \a - 2.5) {};
            \node[draw, vert] (b5) at (3, \a - 2.5) {$\lor$};

            \draw[<->, thick, blue] (0, \a - 0.5) ++(-165:0.9) arc (-160:-15:0.9) coordinate (c);
            \draw (c) ++(0.5, 0) node {\textcolor{blue}{$2^{n^{\varepsilon}}$}};
        \end{scope}

        \pic {common-part};
        \node[fancy-arrow] at (6.25, -1) {};
    \end{scope}

    \begin{scope}[shift = {(0, -3 * \a - 6.3)}]
        \draw[green!50!black, thick, fill = green!50!black, fill opacity = 0.3]
            (1.1, 1.2) -- (2.2, 2.3) to[out = 45, in = 40] (1.3, 2.3) to[out = 220, in = 225] cycle;
        \node[green!50!black] at (1.5, 2) {\small $Y$};

        \draw[red!70!black, thick]
            (1.1, 1) -- (2.2, 2.1) to[out = 45, in = 40] (2.3, 0.9) to[out = 220, in = 225] cycle;
            
        \draw[red!70!black, thick, fill = red!70!black, fill opacity = 0.3]
            (1.4, 1.1) -- (1.8, 1.6) to[out = 45, in = 40] (1.7, 1.1) to[out = 220, in = 225] cycle;

        \begin{scope}[shift = {(9, 0)}]
            \node[draw, vert] (a) at (0, \a - 0.5) {$\land$};

            \node[draw, vert] (b1) at (-3, \a - 1.5) {$\lor$};
            \node[draw, vert] (b2) at (-1.5, \a - 1.5) {$\lor$};
            \node[vert] (b3) at (0, \a - 1.5) {};
            \node[vert] (b4) at (1.5, \a - 1.5) {};
            \node[draw, vert] (b5) at (3, \a - 1.5) {$\lor$};

            \node[draw, rectangle, rounded corners = 2pt, color = green!40!black, fill = green!20]
                (c1) at (-3, \a - 3.5) {\textcolor{black}{CNF}};
            \node[draw, rectangle, rounded corners = 2pt] (c2) at (-1.5, \a - 3.5) {CNF};
            \node[draw, rectangle, rounded corners = 2pt] (c3) at (0, \a - 3.5) {CNF};

            \foreach \i in {1, 2, 3}{
                \draw[thick, ->] (b2) -- (c\i);
            }
        \end{scope}

        \pic {common-part};
    \end{scope}
\end{tikzpicture}
    \caption{Choosing a mirror set}
    \label{fig:steps}
\end{figure}

\begin{remark}[Completeness of the top-down method] \label{rem:complete}
Choosing the mirror set $M$ is the crux of the top-down argument. It follows from linear programming duality that there always exist a mirror set such that any subcircuit that accepts a large enough fraction of $M$ will make a mistake further down the circuit. This means that if we only knew how to construct mirror sets, we could prove a tight lower bound (up to polynomial factors) for constant-depth circuits for any boolean function. We refer to \cite{Hirahara2017,Goos2023} for more formal discussion of the completeness of the top-down method.
\end{remark}

\paragraph{Challenge 2: Block unpredictability.}
Given a CNF $\Gamma$ that accepts a large subset $X\subseteq M$, we would next like to show the existence of a local limit $y\in\Xor^{-1}(0)$ of $X$. To reach a contradiction, it is important that the local limit lies in the target set $Y$; note that $\Gamma$ is only guaranteed to reject inputs in $Y$, not every input in~$\Xor^{-1}(0)$. Our guarantee about $X\subseteq M$ is that every $x\in X$ contains many points in $Y$ at Hamming distance at most $n^{1/3}$. Thus, starting from $x\in X$ and trying to reach a point in $Y$, we may need to flip a whole block of at most $n^{1/3}$ bits. Meir and Wigderson~\cite{Meir2019} already generalised their bit unpredictability lemma to larger blocks of bits (essentially by iteratively applying their lemma for a single bit). However, their generalisation did not supply unpredictable blocks with high enough probability. In this paper, we extend their proof and give a block unpredictability lemma with essentially optimal parameters.

To state our lemma, we first generalise the definition of certificates to whole blocks of bits. Let $X\subseteq \bool^n$ be a set of $n$-bit strings and $R\in \smash{\binom{[n]}{r}}\coloneqq\{ A\subseteq[n]: |A|=r\}$ a block of $r$ coordinates. A \emph{certificate for $R$} (wrt~$X$) is a pair~$(Q,a)$ such that $Q \subseteq [n] \setminus R$, $a \in \{0,1\}^Q$ and there exists $b \in \{0,1\}^R$ such that every~$x \in X$ with $x_Q = a$ satisfies~$x_R \neq b$. In words, the partial assignment defined by $(Q,a)$ will correctly predict that some $r$-bit string $b$ is missing over the coordinates $R$. We prove the following in \cref{sec:block-unpred}.
 
\begin{restatable}[Block unpredictability]{lemma}{Unpred}
\label{thm:unpred}
Let $X\subseteq\{0,1\}^n$ have density $|X|/2^n\geq 2^{-k}$. Then for any $r,q\geq 1$,
\begin{equation} \label{eq:sat-prob}
\Pr_{(\x,\R)\sim X\times  \binom{[n]}{r}}[\, \text{$\x$ contains a size-$q$ certificate for $\R$ wrt $X$}\,] 
~\leq~
O(kqr/n)^{1/6}.
\end{equation}
\end{restatable}

\begin{remark}[Optimality of \cref{thm:unpred}]
We note that \cref{thm:unpred} is optimal in that if $kqr=\Omega(n)$ then there are examples of sets $X$ such that the probability \eqref{eq:sat-prob} is $\Omega(1)$. Let $B_1,\ldots,B_k \subseteq [n]$, $|B_i|=q$, be any pairwise disjoint sets. Consider $X\coloneqq\{x\in\{0,1\}^n :\forall i\in[k],\ \Xor_q(x_{B_i})=0 \}$ of density $|X|/2^n\geq 2^{-k}$. We claim that for every $x\in X$, if we choose $\smash{\R\sim\binom{[n]}{r}}$, then $\Pr[ \text{$x$ contains a size-$q$ certificate for $\R$ wrt $X$}] \geq \Omega(1)$. Indeed, by the birthday paradox, the random block $\R$ will intersect some $B_i$ with probability $\Omega(1)$ and in that case there is a certificate with $Q\coloneqq B_i\setminus \R$ that can predict the parity of $x_{B_i\cap\R}$.
\end{remark}

\paragraph{Challenge 3: Flipping into target set.}
Using block unpredictability, we may now attempt to construct a local limit $y$ of $X\subseteq M$ that lies in the target set $Y$. Suppose we apply the block unpredictability for $X$ to find an input $x\in X$ such that whp over $\R\sim\binom{[n]}{r}$ (where $r\approx n^{1/3}$) we have that any $y$ obtained from $x$ by flipping any subset of bits in $\R$ is a local limit. Can we reach a point in $Y$ via some such block flip?

Let us reformulate this question as follows. Consider the set $\As\coloneqq Y-x \coloneqq\{y-x :y\in Y\}$ (where the arithmetic is that of $\mathbb{Z}_2^n$) of indicator vectors of block-flips that will land us in $Y$ starting from $x$. If we think of $\As$ as a set family $\As \subseteq 2^{[n]}$, then a successful block-flip exists with high probability iff (for $p\coloneqq r/n$)
\begin{equation}\label{eq:sat-def}
\Pr_{\R\sim \binom{[n]}{pn}}[\exists A\in \As\colon A \subseteq \R]
~\geq~ 1 - \epsilon.
\end{equation}
Following Rossman~\cite{Rossman2014}, a set family $\As \subseteq 2^{[n]}$ with property \eqref{eq:sat-def} is called \emph{$(p,\epsilon)$-satisfying}. More generally,~$\As$ is called a \emph{$(p, \epsilon)$-robust sunflower} if, for the \emph{kernel} $K\coloneqq \bigcap_{A\in\As} A$, the family $\{ A\setminus K:A\in\As\}$ is $(p, \epsilon)$-satisfying. An exciting line of work~\cite{Rossman2014,Alweiss2021,Rao2020} has shown that every set family that is ``locally dense''---has many sets of size $\ll r$---contains a robust sunflower for $p=r/n$. These works suggest a strategy for us: we define (for real, this time) our mirror set $M\subseteq \Xor^{-1}(1)$ not as the set of strings $x$ such that $Y-x$ is locally dense, but we instead include in $M$ strings corresponding to kernels of robust sunflowers found inside $Y-x$. That is, instead of including the original locally dense point $x$, we include a nearby point~$x'$ (where $x'-x$ is the kernel) where we also make sure that $x'\in\Xor^{-1}(1)$. This way, we ensure that every $x\in M$ is $(r/n,\epsilon)$-satisfying, and this allows us to find a local limit in $Y$.

A useful sufficient condition for a set family $\As$ to be $(p,\epsilon)$-satisfiable is that no set $I\subseteq[n]$ appears too often as a subset of a randomly chosen $\Ab\sim\As$. Formally, we say $\As$ is \emph{$\kappa$-spread} if for every $I \subseteq [n]$,
\begin{equation}\label{eq:spread}
\Pr_{\Ab\sim\As}[I\subseteq A] ~\leq~ \kappa^{-|I|}.
\end{equation}%
\vspace{-1.5em}
\begin{lemma}[{\cite[Lemma~4]{Rao2020}}]
\label{lem:satisfy}
There exists a universal constant $C > 0$ such that the following holds.
    Suppose $\emptyset \neq \As \subseteq \binom{[n]}{r}$ is $\kappa$-spread for
    $\kappa = (C / p) \log(r / \epsilon)$. Then $\As$ is $(p, \epsilon)$-satisfying.
\end{lemma}

\section{Proof of \cref{thm:main}}

Let $m \coloneqq n^{1/3}$. Suppose for the sake of contradiction that $\Pi$ is a depth-4 circuit of size $|\Pi|\leq 2^{m^{1-\Omega(1)}}$ that computes the $n$-bit $\Xor$ function. We may assume that $\Pi$ is of type $\land\circ\lor\circ\land\circ\lor$, that is, it has an $\land$ gate at the top, the layers below alternate between $\lor$ and $\land$, and there are input literals at the bottom. Starting at the top gate, we take steps down the circuit to reach a contradiction. 

\subsubsection*{First step}
We have $\Xor^{-1}(0) = \bigcup_{i\in[s]} \Sigma_i^{-1}(0)$ where $s\leq |\Pi|$ is the top fanin and $\Sigma_i$ is the $i$-th subcircuit (type $\lor\circ\land\circ\lor$) feeding into the top gate. We now choose any $i\in[s]$ that maximises $|\Sigma_i^{-1}(0)|$, and set $\Sigma\coloneqq\Sigma_i$. In summary,
\begin{itemize}
    \item $\Sigma$ accepts the set~$\Xor^{-1}(1)$,
    \item $\Sigma$ rejects the set $Y\coloneqq\Sigma^{-1}(0)$ of density $|Y|/2^n\geq |\Xor^{-1}(0)|/(2^n|\Pi|)\geq 1/(2|\Pi|) \geq 2^{-m^{1-\Omega(1)}}$. 
\end{itemize}

\subsubsection*{Second step}
We denote the sphere of radius $r$ centered at $x\in\{0,1\}^n$ by
$S_r(x) \coloneqq\{ y\in\{0,1\}^n: \dist(x,y)=r\}$
where $\dist(x,y)$ is the Hamming distance between $x$ and $y$. We consider a radius $r= m$ below.

\begin{claim}\label{clm:sphere}
There is a set $Z\subseteq \{0,1\}^n$, $|Z|\geq|Y|/2$, such that for any $x\in Z$ we have local density
\[\textstyle
|S_m(x)\cap Y|/\binom{n}{m}
~\geq~ (|Y|/2^n)/2.
\]
\end{claim}
\begin{proof}
Write $\alpha\coloneqq |Y|/2^n$ and $\delta_x\coloneqq |S_m(x)\cap Y|/\binom{n}{m}$.
Sample a uniform $\x\sim\{0,1\}^n$ and then sample a uniform $\y\sim S_m(\x)$. Note that~$\y$ is uniform in $\{0,1\}^n$ and hence $\Pr[\y \in Y]\geq \alpha$ and moreover $\Exp[\delta_\x]\geq \alpha$. On the other hand $\Exp[\delta_\x] \leq \alpha/2\cdot\Pr[\delta_\x < \alpha/2] + 1\cdot \Pr[\delta_\x \geq \alpha/2]$. These imply $\Pr[\delta_\x \geq \alpha/2] \geq \alpha/2$.
\end{proof}

Let $Z\subseteq \{0,1\}^n$, $|Z|\geq |Y|/2$, be the set obtained from \cref{clm:sphere} so that~$|S_m(x)\cap Y|/\binom{n}{m}\geq 2^{-m^{1-\Omega(1)}}$ for every $x\in Z$. We will consider each $x\in Z$ in turn and apply the following lemma, which moves $x$ to a nearby input $x'$ (of odd parity) such that we have a satisfying set family surrounding $x'$. Recall that~$Y-x'\coloneqq\{y-x':y\in Y\}$ and we may naturally think of it as a set family $Y-x'\subseteq 2^{[n]}$. We prove the following lemma in \cref{sec:spreadify}.
\begin{lemma}\label{lem:spreadify}
Suppose $x$ is locally dense in $Y$ in that $|S_m(x)\cap Y|/\binom{n}{m} \geq 2^{-m^{1-\Omega(1)}}$. Then there is a center $x'\in\Xor^{-1}(1)$ such that $Y-x'$ is $(m^{1+o(1)}/n,o(1))$-satisfying. Moreover, $\dist(x,x')\leq m^{1-\Omega(1)}$.
\end{lemma}
We define our mirror set to be $M\coloneqq \{x':x\in Z\}\subseteq\Xor^{-1}(1)$ where we obtain each $x'$ by applying \cref{lem:spreadify} to each~$x\in Z$. Each $x'\in M$ could arise from any $x\in Z$ with $\dist(x,x')\leq m^{1-\Omega(1)}$. Hence $|M|/2^n \geq |Z|/(2^n\binom{n}{\leq m^{1-\Omega(1)}}) \geq 2^{-m^{1-\Omega(1)}}$ where we wrote $\binom{n}{\leq r}\coloneqq\sum_{i=0}^r\binom{n}{i}$.

We can now take our second step down the circuit. We have $M\subseteq \Xor^{-1}(1)=\Sigma^{-1}(1)=\bigcup_i \Gamma_i^{-1}(1)$ where $\Gamma_i$ is the $i$-th CNF (type $\land\circ\lor$) feeding into the gate computing $\Sigma$. We choose any $i$ that maximises $|M\cap \Gamma_i^{-1}(1)|$, and set $\Gamma \coloneqq \Gamma_i$. In summary,
\begin{itemize}
    \item $\Gamma$ accepts the set~$X \coloneqq M\cap \Gamma^{-1}(1)$ of density $|X|/2^n\geq |M|/(2^n|\Pi|) \geq 2^{-m^{1-\Omega(1)}}$.
    \item $\Gamma$ rejects the set $Y$.
\end{itemize}

\subsubsection*{Third step}

Fix a constant $\epsilon>0$ such that $|X|/2^n \geq 2^{-m^{1-\epsilon}}$. Apply \cref{thm:unpred} to the set $X$ with parameters
\[
k ~\coloneqq~ m^{1-\epsilon},\quad
r ~\coloneqq~ m^{1+\epsilon/4}, \quad
q ~\coloneqq~ m^{1+\epsilon/2}.
\]
Note that $krq \leq o(n)$ and hence
\begin{align}
\Pr_{(\x,\R)\sim X\times  \binom{[n]}{r}}[\, \text{$\x$ contains a size-$q$ certificate for $\R$ wrt $X$}\,] 
&~\leq~ o(1).
\notag
\intertext{
By averaging, there is some set $X'\subseteq X$, $|X'|\geq |X|/2$, such that for every $x\in X'$,}
\Pr_{\R\sim \binom{[n]}{r}}[\, \text{$x$ contains a size-$q$ certificate for $\R$ wrt $X$}\,]
&~\leq~ o(1).
\label{eq:cert-prob}
\end{align}
Let $x\in X'$ and consider the following process to sample a random variable $\y(x)\in Y\cup \{\bot\}$.
\begin{enumerate}[label=(\roman*),noitemsep]
    \item Sample a random $\R \sim \binom{[n]}{r}$.
    \item \label{it:cert}
    If $x$ contains a size-$q$ certificate for $\R$ wrt $X$, output $\y(x)\coloneqq \bot$ (failure).
    \item \label{it:flip}
    If there is some $y\in Y$ such that $y$ can be obtained from $x$ by flipping some subset of bits in $\R$, output~$\y(x)\coloneqq y$. Otherwise output $\y(x)\coloneqq \bot$ (failure).
\end{enumerate}
Step~\ref{it:cert} fails only with probability $o(1)$ because of \eqref{eq:cert-prob}. Moreover, step~\ref{it:flip} fails with probability $o(1)$ because every $x\in X'\subseteq M$ satisfies the conclusion of \cref{lem:spreadify}. In summary, for every $x\in X'$,
\begin{equation}\label{eq:fail}
\Pr[\,\y(x) = \bot\,] ~\leq~ o(1).
\end{equation}
Let $Y'\coloneqq\bigcup_{x\in X'} \supp(\y(x))\setminus\{\bot\}\subseteq Y$. To estimate the density of $Y'$, note that \cref{eq:fail} implies that each~$x\in X'$ contributes at least one element to $Y'$ and moreover each $y\in Y'$ can result from at most $\smash{\binom{n}{\leq r}}$ many $x\in X'$. Therefore $|Y'|/2^n \geq |X'|/(2^n\binom{n}{\leq r})\geq 2^{-m^{1+\varepsilon/3}}$.

We can now take our third step down the circuit. We have $Y'\subseteq\Gamma^{-1}(0)=\bigcup_i \Lambda_i^{-1}(0)$ where $\Lambda_i$ is the $i$-th clause (type $\lor$) of the CNF $\Gamma$. We choose any $i$ that maximises $|Y'\cap \Lambda_i^{-1}(0)|$, and set $\Lambda \coloneqq \Lambda_i$. In summary,
\begin{itemize}
    \item $\Lambda$ accepts the set~$X$.
    \item $\Lambda$ rejects the set $Y''\coloneqq Y'\cap \Lambda^{-1}(0)$ of density $|Y''|/2^n\geq |Y'|/(2^n|\Pi|)\geq 2^{-m^{1+\varepsilon/2}}=2^{-q}$.
\end{itemize}

\subsubsection*{Final step}

Since the clause $\Lambda$ rejects at least a fraction $|Y''|/2^n\geq 2^{-q}$ of all inputs, it must contain at most $q$ literals. Let $Q\subseteq [n]$, $|Q|\leq q$, be the set of variables mentioned in $\Lambda$. Pick any $y\in Y''$, $x\in X'$, $R\in \binom{[n]}{r}$ such that~$y=(\y(x)\mid \R=R)$ in the above process. Note that $x$ and $y$ differ only on some coordinates in $R$ and~$x$ contains no size-$q$ certificate for $R$ wrt $X$. Thus, there must exist some~$x'\in X$ such that $x'_Q=y_Q$. This means that $\Lambda(x')=\Lambda(y)=0$. But this contradicts the fact that $\Lambda$ accepts all of $X$.

This concludes the proof of \cref{thm:main}.

\subsection{Proof of \cref{lem:spreadify}}
\label{sec:spreadify}

Let $A \coloneqq S_m(x) \cap Y$ so that $|A| / \binom{n}{m} \geq 2^{-m^{1 - \Omega(1)}}$. For notational
simplicity, we assume $x \coloneqq 0^n$. Let $\epsilon = \epsilon(n) \coloneqq 1 / \log n = o(1)$ and $p\coloneqq m/n$. Our goal is to make $A-x=A$ (thought of as a set family) $(1/p)^{1-\epsilon}$-spread by excluding a kernel. Let $I \subseteq [n]$ be the largest set (perhaps $I=\emptyset$) where the $(1/p)^{1-\epsilon}$-spreadness condition~\eqref{eq:spread} fails:
\[
\Pr_{\x \sim A}[\x_I = 1_I]
~>~ p^{(1 - \epsilon) |I|}.
\]
Consider the input~$x'' \in \{0, 1\}^n$ that is the indicator vector for $I$ (that is, $x''_i = 1 \Leftrightarrow i \in I$). This would be our candidate for the center (or kernel) to satisfy the statement of the lemma, except we do not know whether $x'' \in \Xor^{-1}(1)$. To fix this, we flip one more coordinate in $x''$ if needed. Let $i_0 \in [n]\setminus I$ be the coordinate where $1$ occurs most frequently among elements of $\{x\in A: x_I=1_I\}-x''$. Altogether, we define:
\begin{itemize}
\item If $\Xor(x'')=1$, then we set $x' \coloneqq x''$ and $I' \coloneqq I$.
\item If $\Xor(x'')=0$, then we set $x'$ to be $x''$ but with the $i_0$-th bit flipped, and $I' \coloneqq I\cup\{i_0\}$.
\end{itemize}
The following two claims conclude the proof of \cref{lem:spreadify}.

\begin{claim}
$\dist(x,x')=|I'|\leq m^{1-\Omega(1)}$.
\end{claim}
\begin{proof}
For $\x \sim \binom{[n]}{m}$, we have $\Pr[\x_{I} = 1_I]
\ge
\Pr[\x_{I} = 1_I \mid \x \in A]
\Pr[\x \in A]$. Using this we get
\[
2^{-m^{1 - \Omega(1)}}
~=~ \Pr[\x \in A]
~\le~
\frac{\Pr[\x_{I} = 1_I]}{
\Pr[\x_{I} = 1_I \mid \x \in A]
}
~\le~
\frac{p^{|I|}}{p^{(1-\epsilon)|I|}}
~=~ p^{\epsilon |I|}.
\]
Thus $|I'| \le |I| + 1 \le m^{1 - \Omega(1)}/(\epsilon \log (1/p)) + 1 \leq
m^{1 - \Omega(1)}$.
\end{proof}

\begin{claim}
$Y - x'$ is $(m^{1 + o(1)} / n, o(1))$-satisfying. 
\end{claim}
\begin{proof}
Let $B \coloneqq \{x\in A: x_{I'}=1_{I'}\}-x' \subseteq Y-x'$. We will show that $B$ is $(1/p)^{1 - 2.1 \epsilon}$-spread. This, together with \Cref{lem:satisfy}, would imply that~$B$ is $(p^{1 - 3 \epsilon}, o(1))$-satisfying, which proves the claim since $p^{1-3\epsilon}\leq m^{1+o(1)}/n$. Assume for contradiction that~$B$ is not $(1/p)^{1 - 2.1\epsilon}$-spread. This means there is a non-empty $J\subseteq[n]\setminus I'$ with
\[
\Pr_{\x \sim B}[\x_J = 1_J]
~>~
p^{(1 - 2.1 \epsilon) |J|}.
\]
We now claim that $\Pr_{\x \sim A}[\x_{I' \cup J} = 1_{I' \cup J}]>p^{(1-\epsilon)|I'\cup J|}$, which would contradict the maximality of $I$. Indeed, we calculate (assuming $I'=I\cup\{i_0\}$, as the other case is similar)
\begin{align*}
\Pr_{\x \sim A}[\x_{I' \cup J} = 1_{I' \cup J}]
~&=~
\Pr_{\x \sim A}[\x_{I'} = 1_{I'}] \cdot
\Pr_{\x \sim A}[\x_{J} = 1_{J} \mid \x_{I'} = 1_{I'}] \\
&=~ \Pr_{\x \sim A}[\x_{I'} = 1_{I'}] \cdot \Pr_{\x \sim B}[\x_{J} = 1_{J}] \\
&>~ 
\Pr_{\x \sim A}[\x_{I'} = 1_{I'}] \cdot p^{(1 - 2.1 \epsilon) |J|} \\
&\ge~ 
\Pr_{\x\sim A}[\x_{i_0} = 1 \mid \x_{I} = 1_{I}] \cdot \Pr_{\x \sim A}[\x_{I} = 1_{I}]
\cdot p^{(1 - 2.1 \epsilon) |J|} \\
&\ge~ 
\frac{m - |I|}{n - |I|} \cdot p^{(1 - \epsilon) |I|} \cdot
p^{(1 - 2.1 \epsilon) |J|} \\
&\ge~
p \cdot (1 - m^{-\Omega(1)}) \cdot p^{(1 - \epsilon) |I|} \cdot p^{(1 - 2.1 \epsilon) |J|} \\
&\ge~
p^{(1 - \epsilon) (|I| + |J| + 1)} \cdot p^{\epsilon}
\cdot p^{-\epsilon |J|} \cdot (1 - m^{-\Omega(1)}) \cdot p^{-0.1\epsilon} \\
&\ge~
p^{(1 - \epsilon) |I' \cup J|}.
\qedhere
\end{align*}
\end{proof}

\pagebreak
\section{Block unpredictability lemma}
\label{sec:block-unpred}

In this section, we prove the block unpredictability lemma, \Cref{thm:unpred}. We start with some basic facts about  entropy (\cref{sec:entropy}) and set shattering (\cref{sec:shatter}). Then we prove \Cref{thm:unpred} in \cref{sec:unpred-proof}.

\subsection{Entropy}
\label{sec:entropy}

The usual Shannon entropy of a random variable $\X$ is defined by $\entropy(\X) \coloneqq \sum_{x \in \supp(\X)} p(x) \log (1/p(x))$ where~$p(x) \coloneqq \Pr[\x=x]$.
Given two random variables~$\X$ and $\Y$, we define the conditional entropy of~$\X$ given $\Y$ by $\entropy(\X\mid\Y) \coloneqq \Exp_{\bm{y} \sim \Y}\left[\entropy(\X\mid\Y=\bm{y})\right]$. A simple form of the \emph{chain rule} for entropy states that~$\entropy(\X\Y) = \entropy(\X) + \entropy(\Y\mid\X)$.
For convenience, if~$X \subseteq \bool^n$ is a set, we write $\X\sim X$ for the random variable that is uniformly distributed over $X$. In particular, $\X_T$ denotes the random variable $\X$ marginalised onto coordinates $T \subseteq [n]$.
For more background on entropy, we refer to~\cite{Cover2005}. We now recall two useful facts about the entropy of marginal distributions.

\begin{lemma}
\label{lem:deficiency-reduction}
Let $\X\in \{0,1\}^n$ be a random variable with $\entropy(\X) \ge n - k$. For every $t$ and $\delta > 0$,
\[
\Pr_{\T \sim \binom{[n]}{t}}\big[\,\entropy(\X_\T) \ge t - \textstyle\frac{kt}{\delta n}\,\big]
~\ge~
1 - \delta.
\]
\end{lemma}
\begin{proof}
Let ${\bm{\pi}}\colon \mathbb{Z}_n\to \mathbb{Z}_n$ be a uniform random permutation and set $\T_i \coloneqq \{{\bm{\pi}}(i+j) + 1 \mid j \in [t]\}$ for~$i \in [n]$. Note that $\T_i$ and $\T\sim\binom{[n]}{t}$ have the same distribution. By Shearer's inequality, $n - k \le \entropy(\X) \le 1/t\cdot  \sum_{i \in [n]} \Exp_{\bm{\pi}}[\entropy(\X_{\T_i})] = n/t\cdot \Exp_\T[\entropy(\X_\T)]$, and therefore $\Exp_\T[\entropy(\X_\T)] \ge t - kt/n$. Applying Markov's inequality to the nonnegative random variable $t - \entropy(\X_\T)$ completes the proof.
\end{proof}

\begin{lemma}
\label{lem:conditional-entropy}
Let $\X\in \{0,1\}^n$ be a random variable with $\entropy(\X) \ge n - k$. For every $T\in\binom{[n]}{t}$ and $\delta > 0$,
\[
\Pr_{\x\sim\X}\big[\entropy(\X_T \mid \X_{[n] \setminus T} = \x_{[n] \setminus T}) \geq t - k/\delta\,\big]
~\geq~ 1 - \delta.
\]
\end{lemma}
\begin{proof}
We have $n - k \le \entropy(\X) = \entropy(\X_{[n] \setminus T}) + \entropy(\X_T \mid \X_{[n] \setminus T})$ by the chain rule. Using $\entropy(\X_{[n] \setminus T}) \le n - t$ we get $\entropy(\X_T \mid \X_{[n] \setminus T}) \ge t - k$. Since $\entropy(\X_T \mid \X_{[n] \setminus T}) = \Exp_{\x\sim \X} \left[\entropy(\X_T \mid \X_{[n] \setminus T} = \x_{[n] \setminus T})\right]$ by Markov's inequality we get the desired probability.
\end{proof}

\subsection{Shattered sets}
\label{sec:shatter}

Let $\As\subseteq 2^{[n]}$ be a set family. For $B\subseteq[n]$, we write $\As_B=\{A\cap B : A\in \As\}$ for the projection of $\As$ onto~$B$. We say that $\As$ \emph{shatters}~$B \subseteq [n]$ if $\As_B=2^B$. Similarly, for a set of strings $X\subseteq \{0,1\}^n$ we say that $X$ \emph{shatters} $B \subseteq [n]$ if the set family~$\{\{i \in [n] \mid x_i = 1\} \mid x \in X\}$ shatters $B$. We will use the following strong version of the usual Sauer--Shelah lemma.
\begin{lemma}[\cite{Pajor1985}]
\label{lem:sauer-shelah}
   A set $X \subseteq \{0,1\}^n$ shatters at least $|X|$ sets.
\end{lemma}

We say that a set family $\As$ is \emph{downward-closed} if for every pair of sets $S$, $T$ such that $S \subseteq T$, $T \in \As$ we also have $S \in \As$.
We will use the following simple combinatorial fact, which states that the density of $\mathcal{A}$ decreases monotonically over the slices $\smash{\binom{[n]}{k}}$.
\begin{lemma}
\label{lem:increasing-density}
Let $\mathcal{A} \subseteq 2^{[n]}$ be downward-closed and define $\mathcal{A}_k \coloneqq \mathcal{A} \cap \binom{[n]}{k}$. Then
\[
\frac{|\mathcal{A}_1|}{\binom{n}{1}} \ge \frac{|\mathcal{A}_2|}{\binom{n}{2}} \ge \dots \ge \frac{|\mathcal{A}_n|}{\binom{n}{n}}.
\]  
\end{lemma}
\begin{proof}
Consider the bipartite graph $G=(\mathcal{A}_k\cup \mathcal{A}_{k-1}, E)$ where $\{u,v\} \in E$ iff $u \supseteq v$. The degree in the first part equals $k$ and the degree in the second part is at most $n - k + 1$. Hence $|\mathcal{A}_k| k = |E| \le (n-k+1) |\mathcal{A}_{k-1}|$. Thus $|\mathcal{A}_k|/|\mathcal{A}_{k-1}|\le (n-k+1)/k$. Since $\binom{n}{k}/\binom{n}{k-1} = (n-k+1)/k$, the lemma follows.
\end{proof}

\begin{lemma}
\label{lem:shattered}
Let $X \subseteq \{0,1\}^n$ have density $|X|/2^n \ge 2^{-k}$. Then for any $r\geq 1$,
\[
\Pr_{\R \sim \binom{[n]}{r}}[\,\text{$X$ does not shatter $\R$}\,]
~\leq~
O(kr/n)^{1/4}.\]
\end{lemma}
\begin{proof}
Let $\mathcal{A} \subseteq 2^{[n]}$ be the family of subsets of $[n]$ shattered by $X$. Observe that $\mathcal{A}$ is downward-closed. By \Cref{lem:sauer-shelah}, $|\mathcal{A}| \ge |X| \ge 2^{n-k}$. Our goal is to show that $\mathcal{A}$ contains a random set of size $r$ with high probability. Let $\delta>0$ be a parameter to be chosen later. For $U\in\binom{[n]}{10r}$ we define
\[
\text{$U$ is \emph{good}}
~\iff~
|\As_U|\geq 2^{10r - 10kr/(\delta n)}.
\]
If we choose $\bm{U} \sim \binom{[n]}{10 r}$ at random, then \Cref{lem:deficiency-reduction} implies that $\entropy(\bm{\mathcal{A}}_{\bm{U}}) \ge 10r - 10kr/(\delta n)$ (where we identified~$\mathcal{A}$ with a subset of $\bool^n$) with probability $\geq 1 - \delta$. Since entropy is at most the log of size,
\[\textstyle
\Pr_{\bm{U} \sim \binom{[n]}{10 r}}[\,\text{$\Ub$ is good}\,]
~\geq~ 1-\delta.
\]
\vspace{-1em}
\begin{claim}
Suppose $U\in \binom{[n]}{10 r}$ is good. Then $\tau\coloneqq \Pr_{\R\sim\binom{U}{r}}[\R\in\As]\geq 1- 10kr/(\delta n)-2^{-5r}$. 
\end{claim}
\begin{proof}
There are $\tau \binom{10r}{r}$ many size-$r$ subsets in $\As_U$. Because $\As_U$ is downward-closed, we get from \Cref{lem:increasing-density} that $|\As_U| \leq \sum_{i = 0}^{r-1} \binom{10r}{i} + \tau \sum_{i=r}^{10r} \binom{10r}{i}$. Note that $\sum_{i = 0}^{r-1} \binom{10r}{i} \le 2^{10r \cdot h(1/10)} \le 2^{5r}$ where $h$ is the binary entropy function. Thus $2^{10r - 10kr/(\delta n)}\leq|\As_U|\leq 2^{5r} + \tau2^{10r}$ and hence
\[
\tau
~\ge~
(2^{10r - \frac{10kr}{\delta n}} - 2^{5r})/2^{10r}
~=~
2^{-\frac{10kr}{\delta n}} - 2^{-5r}
~\ge~
1 - 10kr/(\delta n) - 2^{-5r}.
\qedhere
\]
\end{proof}

Sampling $\R\sim\binom{[n]}{r}$ is equivalent to first sampling $\Ub\sim\binom{[n]}{10r}$ and then sampling $\R \sim \binom{\bm{U}}{r}$. Thus,
\[
p~\coloneqq~
\Pr_\R[\R \notin \As]
~\le~
\Pr_\Ub[\text{$\Ub$ is bad}]
+\Pr_{\R \sim \binom{\bm{U}}{r}}[\R \notin \As \mid \text{$\Ub$ is good}]
~\le~
\delta + (1 - \tau).
\]
Choosing $\delta \coloneqq \sqrt{10 k r/n}$ we get that $p \le 2 \sqrt{10 k r/n} + 2^{-5r}$. Let us analyse two cases.
\begin{itemize}
\item Case $r \ge 0.1 \log (n/k) $: Here $2^{-5r} \le \sqrt{k/n}$, so we get $p \le 3 \sqrt{10rk/n}$.
\item Case $r < 0.1 \log (n/k)$: Here we observe that by \Cref{lem:increasing-density} we can bound $p$ with $p' \coloneqq \Pr[\bm{I} \not\in \mathcal{A}]$ where $\bm{I}$ is a random set of size $0.1 \log (n/k)$. From the previous case we have $p' \le 3 \sqrt{k/n \cdot \log (n/k)}$. There exists a constant $C$ such that $k/n \cdot\log (n/k) \le C \cdot \sqrt{k/n}$ (here we use $(\log x)/x = O(1/\sqrt{x})$ for $x \coloneqq n/k$), so $p \le p' \le \sqrt{C} (k/n)^{1/4} \le O(kr/n)^{1/4}$. 
\end{itemize}
\end{proof}

\subsection{Proof of \Cref{thm:unpred}}
\label{sec:unpred-proof}

\Unpred*
\begin{proof}
Let $\delta> 0$ and $t\geq r$ be parameters to be chosen later. Consider sampling $\x\sim X$, $\smash{\T \sim \binom{[n]}{t}}$, and then $\R \sim \binom\T{r}$. Note that $\R\sim\binom{[n]}{r}$. Let $E$ be the event ``$\x$ contains a size-$q$ certificate for $\R$ (wrt~$X$)'' and let $E'$ be the event ``$\x$ contains a size-$q$ certificate $(Q,x_Q)$ for $\R$ such that $Q \cap \T = \emptyset$.''

\begin{claim}
    \label{cl:certificate-separation}
    $\Pr[E' \mid E] \geq 1 - q t/n$.
\end{claim}
\begin{proof}
    Fix $x \in \bool^n$ and $R \in \binom{[n]}{r}$.  Let $Q_{R,x}\in\binom{[n]}{q}$ be the lexicographically first (if any) subset such that $(Q_{R,x}, x_{Q_{R,x}})$ is a certificate for $R$. Then if $Q_{R,x}$ exists we have 
    $\Pr[Q_{R,x} \cap \T \neq \emptyset] \le \sum_{i \in Q_{R,x}} \Pr [i \in \T] \le qt/n$.
    Then by the law of total probability we get $\Pr[Q_{\R, \x} \cap \T \neq \emptyset \mid E] \le q t/n$. Since~$E'$ is implied by $E$ and the negation of the event ``$Q_{\R, \x} \cap \T \neq \emptyset$,'' the claim follows.
\end{proof}

To prove the lemma, let us first branch on $E'$:
\(\Pr[E] = \Pr[E' \land E] + \Pr[\lnot E' \land E]\).
Using \Cref{cl:certificate-separation}, we bound the second term:
\(\Pr[\lnot E' \land E] \leq \Pr[E] \cdot \Pr[\lnot E' \mid E] \leq q t/ n.\)
Next we bound the first term $\Pr[E' \land E] = \Pr[E']$. Let $L$ be the event ``$\entropy(\X_\T \mid \X_{[n] \setminus \T} = \x_{[n] \setminus \T}) \geq t - k/\delta$''. By \Cref{lem:conditional-entropy} we have $\Pr[\lnot L] \leq \delta$. We can now apply \Cref{lem:shattered} to see that
\begin{equation*}
\Pr[E']
~\leq~
\Pr[\lnot L] + \Pr[\,\text{$\supp(\X_\T \mid \X_{[n] \setminus \T} = \x_{[n] \setminus \T})$ does not shatter $\R$} \mid L ]
~\leq~
\delta + O(kr/(\delta t))^{1/4}.
\end{equation*}
The overall probability that $\x$ contains a certificate for $\R$ is therefore $\Pr[E]\leq q t/n + \delta + O(kr/(\delta t))^{1/4}$. To minimize $\delta + O(kr/(\delta t))^{1/4}$, we pick $\delta \coloneqq (kr/t)^{1/5}$ and get 
\(\Pr[E] \leq q t/n + O(k r/t)^{1/5}.\) 
Now picking $t \coloneqq \Theta((kr)^{1/6} (n/q)^{5/6})$ we get \(\Pr[E] \leq O(k q r/n)^{1/6}.\) Let us finally verify the requirement $t \ge r$:
\[
t/r
~=~ \Theta((kr)^{1/6} (n/q)^{5/6} \cdot r^{-1})
~=~ \Theta(k \cdot (n/(kqr))^{5/6})
~\ge~ \Omega(n/(kqr))^{5/6}.
\]
If the RHS is $<1$, the lemma is trivially true. Otherwise it is $\geq 1$ and the requirement is satisfied.
\end{proof}

\medskip
\subsection*{Acknowledgements}
We thank the anonymous FOCS reviewers for their comments and Emanuele Viola for references.
This work was supported by the Swiss State Secretariat for Education, Research and Innovation (SERI) under contract number MB22.00026.

\medskip

\DeclareUrlCommand{\Doi}{\urlstyle{sf}}
\renewcommand{\path}[1]{\small\Doi{#1}}
\renewcommand{\url}[1]{\href{#1}{\small\Doi{#1}}}

\bibliographystyle{alphaurl}
\bibliography{references}

\end{document}